\documentclass{article}%
\usepackage{amsfonts}
\usepackage{amsmath}
\usepackage{amssymb}
\usepackage{graphicx}%
\newtheorem{theorem}{Theorem}[section]

\newtheorem{corollary}[theorem]{Corollary}

\newtheorem{lemma}[theorem]{Lemma}

\newtheorem{proposition}[theorem]{Proposition}
\newtheorem{remark}[theorem]{Remark}

\newenvironment{proof}[1][Proof]{\noindent\textbf{#1.} }{\ \rule{0.5em}{0.5em}}
\numberwithin{equation}{section}
\begin{document}

\title{On the Global Structure of Normal Forms for Slow-Fast Hamiltonian Systems}
\author{M. Avenda\~{n}o Camacho and Yu. Vorobiev\\
Deparment of Mathematics, University of Sonora\\
Rosales and Blvd. Luis Encinas Hermosillo, M\'exico, 83000\\
E-mail: yuvorob@gmail.com}

\maketitle

\begin{abstract}
In the framework of Lie transform and the global method of averaging, the normal forms of a multidimensional slow-fast Hamiltonian system are studied in the case when the flow of the unperturbed (fast) system is periodic and the induced $\mathbb{S}^{1}$-action  is not necessarily free and trivial. An intrinsic splitting of the second term in a $\mathbb{S}^{1}$-invariant normal form of first order is derived in terms of the Hannay-Berry connection associated with the periodic flow.

\end{abstract}

\section{Introduction}

In this paper, in the context of normal form, we deal with a class of
so-called slow-fast Hamiltonian systems \cite{ArKN-88,Ne-08} of the
form
\begin{equation}
\dot{y}=-\frac{\partial H}{\partial x}, \ \ \dot{x}=\frac{\partial
H}{\partial y}, \label{SF1}
\end{equation}%
\begin{equation}
\dot{p}=-\varepsilon\frac{\partial H}{\partial q}, \ \ \dot
{q}=\varepsilon\frac{\partial H}{\partial p}, \label{SF2}%
\end{equation}
where $(y,x)\in\mathbb{R}^{2r},$ $(p,q)\in\mathbb{R}^{2k}$  and $\varepsilon$
is a small perturbation parameter. System (\ref{SF1}), (\ref{SF2}) is
Hamiltonian relative to a function $H=H(p,q,y,x)$ and the $\varepsilon
$-dependent Poisson bracket $\{,\}=\{,\}_{0}+\varepsilon\{,\}_{1}$ on
$\mathbb{R}^{2r}\times\mathbb{R}^{2k}$. From the viewpoint of Hamiltonian
perturbation theory, this splitting of the Poisson bracket into a slow and a
fast part leads to the following unusual feature of the perturbed model: the
unperturbed (fast) system and the perturbation are Hamiltonian relative to
the different (nonisomorphic) Poisson structures. In this situation, one can
expect that the perturbation effects are not only representated by correction
terms in the Hamiltonian but also related to the rescaling of Poisson brackets.

We assume that the flow of the unperturbed system is periodic and
hence induces an $\mathbb{S}^{1}$-action which is canonical relative to the
fast Poisson bracket $\{,\}_{0}$. For $r >  1$, such a situation occurs in the
case when the unperturbed motion is described by a family of systems which are
superintegrable in the noncommutative sense \cite{Nekh88}. The normalization
question comes from the fact that the $\mathbb{S}^{1}$-action does not respect
the perturbation vector field of system (\ref{SF1}), (\ref{SF2}) since the
slow Poisson bracket $\{,\}_{1}$ is not $\mathbb{S}^{1}$-invariant in
general. The traditional  averaging procedure \cite{ArKN-88,Ne-08}
works within domains of (generalized) action-angles variables where the
original $\mathbb{S}^{1}$-action is trivial. We are interested in the global
structure of normal forms for system (\ref{SF1}), (\ref{SF2}) in the general
case when the $\mathbb{S}^{1}$-action is not necessarily free and trivial. Our
approach is based on the global averaging technique on $\mathbb{S}^{1}%
$-manifolds \cite{MisVor-12,Cush-84,Mos-70} which refers to
the flow on a phase space  rather than  to a local coordinate description.
One of the important tools here is the Hannay-Berry connection \cite{MaMoRa-90,Mon-88}
associated to the $\mathbb{S}^{1}$-action which naturally arises in the averaging procedure for symplectic and Poisson structures
\cite{Vor-08,Vor-11,VorMis-12}. We show that the Hamiltonian
vector field of system (\ref{SF1}), (\ref{SF2}) can be transformed by a
near-identity mapping to an $\mathbb{S}^{1}$-invariant normal form of first
order whose second term splits with respect the Hannay-Berry connection into
two parts $P_{\operatorname{hor}}$ and $P_{\operatorname{ver}}$  with the
following properties. The vertical component $P_{\operatorname{ver}}$ is a
Hamiltonian vector field relative to the fast Poisson bracket $\{,\}_{0}$ and
an $\mathbb{S}^{1}$-invariant function which interpreted as a first correction
to the Hamiltonian $H$. This interpretation is motivated by the fact that
system (\ref{SF1}), (\ref{SF2}) can be approximated by a Hamiltonian system with
$\mathbb{S}^{1}$-symmetry on a phase space equipped with a corrected Poisson
bracket. The horizontal component $P_{\operatorname{hor}}$ involves the
horizontal lift of the Poisson tensor on the slow $(p,q)$-space which
satisfies the Jacobi identity only in the case when the curvature of the
Hannay-Berry connection is zero. Therefore, in general, $P_{\operatorname{hor}%
}$  does not inherit any natural Hamiltonian structure. These results are
applied to the construction of approximate first integrals of system
(\ref{SF1}), (\ref{SF2}) and illustrated by some examples.

\section{Averaging and Integrating Operators}

In this section, we collect some facts concerning algebraic properties of
the averaging procedure on general $\mathbb{S}^{1}$-manifolds. For more
details, see, for example, \cite{MisVor-12,Cush-84,Mos-70}.

Suppose that on a manifold $M$ we are given a complete vector field $\Upsilon$
with $2\pi$-periodic flow, $\operatorname{Fl}_{\Upsilon}^{t+2\pi
}=\operatorname{Fl}_{\Upsilon}^{t}$. Then, we have an action on $M$ of the
circle $\mathbb{S}^{1}=\mathbb{R}\diagup2\pi\mathbb{Z}$ with infinitesimal
generator $\Upsilon$. \ Let us associated to this $\mathbb{S}^{1}$-action
\ the following operations. Denote by $\mathcal{T}_{s}^{k}(M)$ be the space of
all tensor fields on $M$ of type $(k,m)$ and by $\mathcal{L}_{X}%
:\mathcal{T}_{s}^{k}(M)\rightarrow\mathcal{T}_{s}^{k}(M)$ the Lie derivative
along $\Upsilon$. For every tensor field $A\in\mathcal{T}_{s}^{k}(M)$, its
average with respect to the $\mathbb{S}^{1}$-action is a tensor field  $\langle A\rangle$
$\in\mathcal{T}_{s}^{k}(M)$ of the same type which is defined by
\begin{equation}
\langle A\rangle:=\frac{1}{2\pi}\int_{0}^{2\pi}(\operatorname{Fl}_{\Upsilon}^{t})^{\ast
}Adt. \label{AV1}%
\end{equation}
This formula gives the global averaging operator associated to the $\mathbb{S}^{1}%
$-action on $M$. A tensor field $A\in\mathcal{T}_{s}^{k}(M)$ is said to be
invariant with respect to the $\mathbb{S}^{1}$-action if $(\operatorname{Fl}%
_{\Upsilon}^{t})^{\ast}A=A$  $(\forall t\in\mathbb{R})$ or, equivalently,
$\mathcal{L}_{\Upsilon}A=0$.  In terms of the $\mathbb{S}^{1}$-average of $A$
the $\mathbb{S}^{1}$-invariance condition reads $ A=\langle A\rangle$.

Introduce also the $\mathbb{R}$-linear operator $\mathcal{S}:\mathcal{T}_{s}^{k}(M)\rightarrow\mathcal{T}_{s}^{k}(M)$ given by%
\begin{equation}
\mathcal{S}(A):=\frac{1}{2\pi}\int_{0}^{2\pi}(t-\pi)(\operatorname{Fl}%
_{\Upsilon}^{t})^{\ast}Adt. \label{AV2}%
\end{equation}
Then, we have the following important algebraic identities involving the
operators $\mathcal{L}_{\Upsilon},\langle\ \rangle$ and $\mathcal{S}$ \cite{MisVor-12}.
\begin{lemma}
For every $A\in\mathcal{T}_{s}^{k}(M)$, the following identities hold%
\begin{equation}
\mathcal{L}_{\Upsilon}\circ\mathcal{S}(A)=A-\langle A\rangle, \label{Pr1}%
\end{equation}%
\begin{equation}
\langle\mathcal{L}_{\Upsilon}(A)\rangle=\mathcal{L}_{\Upsilon}\langle A\rangle=0, \label{Pr2}%
\end{equation}%
\begin{equation}
\langle\mathcal{S}(A)\rangle=\mathcal{S}(\langle A\rangle)=0. \label{Pr3}%
\end{equation}
\end{lemma}

For a given $A$, one can think of (\ref{Pr1}) as a homological equation
involving the Lie derivative along $\Upsilon$. Then, this equation admits a
solution with zero average of the form $\mathcal{S}(A)$. In this context, it
is natural to call $\mathcal{S}$ an integrating operator.

Remark also that operators (\ref{AV1}) and (\ref{AV2}) are well-defined on
the exterior algebras of multilvector fields and differential forms on $M$.
Together with the Lie derivative, these operators are natural with respect to
the exterior derivative $d$ on $M$, that is,
\[
\mathcal{S}(d\omega)=d(\mathcal{S}(\omega))\text{ and }d(\langle\omega\rangle)=\langle d\omega\rangle
\]
for any $k$-form $\omega$ on $M$. Moreover, we have the similar
properties with respect to the interior product. Recall that the interior
product of a 1-form $\alpha$ and $k$-vector field $A$ on $M$ is a
$(k-1)$-vector field $\mathbf{i}_{\alpha}A$ defined by $(\mathbf{i}_{\alpha
}A)(\alpha_{1},...,\alpha_{k-1})=A(\alpha,\alpha_{1},...,\alpha_{k-1}).$ If
$\alpha$  is an $\mathbb{S}^{1}$-invariant 1-form, then
\[
\langle\mathbf{i}_{\alpha}A\rangle=\mathbf{i}_{\alpha}\langle A\rangle\text{ and }\mathcal{S}%
(\mathbf{i}_{\alpha}A)=\mathbf{i}_{\alpha}\mathcal{S}(A)
\]
for an arbitrary $k$-vector field $A$.

\section{Setting of the Problem}

Consider the phase space $\mathbb{R}_{y,x}^{2r}\times\mathbb{R}_{p,q}^{2k}$
endowed with the $\varepsilon$-dependent Poisson bracket
\begin{equation}
\{,\}=\{,\}_{0}+\varepsilon\{,\}_{1}\text{,} \label{PB}%
\end{equation}
where $\{,\}_{0}$ and $\{,\}_{1}$ denote the natural lifts of the canonical
Poisson brackets on the factors $\mathbb{R}_{y,x}^{2r}$ and $\mathbb{R}%
_{p,q}^{2k}$, respectively. Suppose we start with slow-fast Hamiltonian
system (\ref{SF1}), (\ref{SF2}) associated to a smooth function $H=H(p,q,y,x)$%
. As was mentioned, this system is Hamiltonian relative to Poisson bracket
(\ref{PB}) and the function $H$. The corresponding Hamiltonian vector field
$X_{H}$ is represented as follows
\[
X_{H}=X_{H}^{(0)}+\varepsilon X_{H}^{(1)},
\]
where the unperturbed vector field $X_{H}^{(0)}$ and the perturbation vector
field $X_{H}^{(1)}$ are Hamiltonian relative to $H$ and the Poisson brackets
$\{,\}_{0}$ and $\{,\}_{1}$, respectively,
\begin{equation}
X_{H}^{(0)}=-\frac{\partial H}{\partial x}\cdot\frac{\partial}{\partial
y}+\frac{\partial H}{\partial y}\cdot\frac{\partial}{\partial x}, \label{PH1}%
\end{equation}%
\begin{equation}
X_{H}^{(1)}=-\frac{\partial H}{\partial q}\cdot\frac{\partial}{\partial
p}+\frac{\partial H}{\partial p}\cdot\frac{\partial}{\partial q}. \label{PH2}%
\end{equation}
We assume that the unperturbed system admits an invariant open domain
$M\subseteq$ $\mathbb{R}_{y,x}^{2r}\times\mathbb{R}_{p,q}^{2k}$ such that the
flow $\operatorname{Fl}_{X_{H}^{(0)}}^{t}$ of $X_{H}^{(0)}$ is periodic on $M$
with frequency function $\omega\in C^{\infty}(M),$ $\omega \rangle  0$. This means
that $\operatorname{Fl}_{X_{H}^{(0)}}^{t+T(m)}(m)=\operatorname{Fl}%
_{X_{H}^{(0)}}^{t}(m)$ for all $t\in\mathbb{R}$ and $m\in M$. Here
$T=\frac{2\pi}{\omega}$ is the period function. Then, the flow of the vector
field
\begin{equation}
\Upsilon:=\frac{1}{\omega}X_{H}^{(0)} \label{ING1}%
\end{equation}
is $2\pi$-periodic and hence $\Upsilon$ is an infinitesimal generator of the
$\mathbb{S}^{1}$-action on $M$.

It is clear that the frequency function $\omega$ and the Hamiltonian $H$ are
$\mathbb{S}^{1}$-invariant. Moreover, by the period-energy relation
\cite{BaSn-92,Gor-69} for periodic Hamiltonian flows, we have the equality
\begin{equation}
(d_{y}H+d_{x}H)\wedge(d_{y}\omega+d_{x}\omega)=0, \label{PE}%
\end{equation}
where $d_{y}$ and $d_{x}$ denote the partial exterior derivatives on $M$
with respect to the fast variables $y$ and $x$, respectively. Relation
(\ref{PE}) means that, for a fixed $(p,q)$, the frequency function $\omega$ is
constant along the intersection of a level set of $H$ and the slice
$\mathbb{R}_{y,x}^{2r}\times\{(p,q)\}$. It is also easy to see from relation
(\ref{PE}) that the $\mathbb{S}^{1}$-action is canonical with respect to the
bracket $\{,\}_{0}$. On the other hand, as we will show below (see Lemma \ref{lemm7}),
the $\mathbb{S}^{1}$-action does not preserve the slow Poisson bracket
$\{,\}_{1}$, in general. Therefore, the perturbation vector field
$X_{H}^{(1)}$ is not necessarily $\mathbb{S}^{1}$-invariant. This fact rises
the normalization question: in the class of near-identity mappings on $M$,
bring the Hamiltonian vector field $X_{H}$ to an $\mathbb{S}^{1}$-invariant
normal form of desired order in $\varepsilon$.

\section{The Hannay-Berry Connection Associated to the $\mathbb{S}^1$-Action}

To formulate our main results, we need some preliminary facts related to the
averaging procedure on phase spaces with $S^{1}$-symmetry, \cite{MisVor-12,GoKnMa,MaMoRa-90,Mon-88,Vor-11,VorMis-12}.

Throughout this section, we will use operators the $\langle$ $ \rangle  $ and $\mathcal{S}$ in
(\ref{AV1}) and (\ref{AV2}) which are associated to the $\mathbb{S}^{1}%
$-action with infinitesimal generator (\ref{ING1}).

\begin{lemma}\label{lem1}
The $\mathbb{S}^{1}$-action associated to the periodic flow of $X_{H}^{(0)}$
is Hamiltonian relative to the fast Poisson bracket $\{,\}_{0}$,
\begin{equation}
\Upsilon=X_{J}^{(0)}, \label{AH1}%
\end{equation}
where the momentum map $J\in C^{\infty}(M)$ is given by
\begin{equation}
J=\frac{1}{\omega}\mathbf{i}_{X_{H}^{(0)}} \langle ydx \rangle  . \label{AH2}%
\end{equation}
Moreover,
\begin{equation}
 \langle \frac{\partial J}{\partial p^{i}} \rangle  = \langle \frac{\partial J}{\partial q^{i}} \rangle  =0
\label{AH3}%
\end{equation}
for $i=1,...,k$.
\end{lemma}

\begin{proof}
Let $\eta=ydx$. Then, $J=\mathbf{i}_{\Upsilon} \langle \eta \rangle  $ and
\[
\mathbf{i}_{X_{H}^{(0)}}d\eta=-dH+d_{p}H+d_{q}H.
\]
Using property (\ref{Pr2}), the $\mathbb{S}^{1}$-invariance of $dH$ and
Cartan's formula, we get
\begin{align*}
dJ  &  =d(\mathbf{i}_{\Upsilon} \langle \eta \rangle  )=\mathcal{L}_{\Upsilon} \langle \eta
 \rangle  -\mathbf{i}_{\Upsilon} \langle d\eta \rangle  \\
&  =-\frac{1}{\omega} \langle \mathbf{i}_{X_{H}^{(0)}}d\eta \rangle =\frac{1}{\omega}\left(  dH- \langle d_{p}H \rangle  - \langle d_{q}H \rangle  \right).
\end{align*}
From here, taking into account that the 1-forms $dp^{i},dq^{i}$ are
$\mathbb{S}^{1}$-invariant, we deduce the relations%
\begin{equation}
d_{y}J=\frac{1}{\omega}d_{y}H, \ \ d_{x}J=\frac{1}{\omega}d_{x}H,
\label{Rel1}%
\end{equation}%
\begin{equation}
d_{p}J=\frac{1}{\omega}\left(  d_{p}H- \langle d_{p}H \rangle  \right)  , \ \ 
d_{q}J=\frac{1}{\omega}\left(  d_{q}H- \langle d_{q}H \rangle  \right)  \label{Rel2}%
\end{equation}
which imply (\ref{AH1}) and (\ref{AH3}).
\end{proof}

As a consequence of (\ref{PE}) (\ref{Rel1}) and (\ref{Rel2}), we get the
following fact.

\begin{corollary}
The differentials $dH$ and $dJ$ are linear independent on $M$ if and only if%
\begin{equation}
 \langle d_{p}H \rangle  + \langle d_{q}H \rangle  \neq0. \label{Cr}%
\end{equation}

\end{corollary}

\begin{remark}
If the $\mathbb{S}^{1}$-action is free on $M$, then formula (\ref{AH2}) gives
the standard action along the periodic orbits of $X_{H}^{(0)}$ \cite{Arno-63}.
\end{remark}

Now, using the momentum map $J$ and operator (\ref{AV2}), we define the 1-form
$\Theta=\Theta_{i}^{p}dp^{i}+\Theta_{i}^{q}dq^{i}$ on $M$ with coefficients
\begin{equation}
\Theta_{i}^{p}:=\mathcal{S}(\frac{\partial J}{\partial p^{i}}),\text{
\ }\Theta_{i}^{q}:=\mathcal{S}(\frac{\partial J}{\partial q^{i}}),
\label{Hor2}%
\end{equation}
It follows from property (\ref{Pr3}) that
\begin{equation}
 \langle \Theta_{i}^{p} \rangle  = \langle \Theta_{i}^{q} \rangle  =0 \label{AZ}%
\end{equation}
for $i=1,...,k$. Here and throughout the remainder of the text, the\ summation
on repeated indices will be understood.

Consider the pre-symplectic 2-form $dy\wedge dx$ on $\mathbb{R}^{2r}%
\times\mathbb{R}^{2k}$ associated to the fast Poisson bracket $\{,\}_{1}$. The
following lemma shows that the differential of the 1-form $\Theta$ measures
the deviation of $dy\wedge dx$ from the property of being invariant with
respect to the $\mathbb{S}^{1}$-action.

\begin{lemma}
The average $\mathbb{S}^{1}$-average of the 2-form $dy\wedge dx$ has
the following representation on $M$:
\begin{equation}
 \langle dy\wedge dx \rangle  =dy\wedge dx-d\Theta. \label{FS1}%
\end{equation}
\end{lemma}
\begin{proof}
First, we observe that the closed 2-form $\sigma=dy\wedge dx$ satisfies the
relation
\begin{equation}
\sigma= \langle \sigma \rangle  +\text{ }d\circ\mathbf{i}_{\Upsilon}\mathcal{S}(\sigma)
\label{FS2}%
\end{equation}
Indeed, property (\ref{Pr1}) together with Cartan's formula yields
\[
d\circ\mathbf{i}_{\Upsilon}\mathcal{S}(\sigma)=\mathcal{L}_{\Upsilon
}(\mathcal{S}(\sigma))=\sigma- \langle \sigma \rangle
\]
By (\ref{AH1}) we have
\[
\mathbf{i}_{\Upsilon}\sigma=-(d_{y}J+d_{x}J)=-dJ+(d_{p}J+d_{q}J)
\]
and consequently,
\[
\mathbf{i}_{\Upsilon}\mathcal{S}(\sigma)=\mathcal{S}(\mathbf{i}_{\Upsilon
}\sigma)=-\mathcal{S}(dJ)+\mathcal{S}(d_{p}J+d_{q}J)=-dJ+\Theta.
\]
Putting this equality into (\ref{FS2}), we get (\ref{FS1}).
\end{proof}

Introduce now the following vector fields on $M$:
\begin{equation}
\operatorname{hor}_{i}^{p}:=\frac{\partial}{\partial p^{i}}+X_{\Theta_{i}^{p}%
}^{(0)},\text{ \ \ }\operatorname{hor}_{i}^{q}:=\frac{\partial}{\partial
q^{i}}+X_{\Theta_{i}^{q}}^{(0)}. \label{Hor1}%
\end{equation}

\begin{lemma}\label{lem2}
The following identities hold
\begin{equation}
\mathcal{L}_{\operatorname{hor}_{i}^{p}}J=\mathcal{L}_{\operatorname{hor}%
_{i}^{q}}J=0, \label{Hor3}%
\end{equation}%
\begin{equation}
[\operatorname{hor}_{i}^{p},\Upsilon]=[\operatorname{hor}_{i}%
^{q},\Upsilon]=0, \label{Hor4}%
\end{equation}
for all $i=1,...,k$.
\end{lemma}
\begin{proof}
Definition (\ref{Hor2}) and property (\ref{Pr1}) imply
that
\[
\mathcal{L}_{\Upsilon}\Theta_{i}^{p}=\frac{\partial J}{\partial p^{i}},\text{
\ }\mathcal{L}_{\Upsilon}\Theta_{i}^{q}=\frac{\partial J}{\partial q^{i}}.
\]
Moreover, property (\ref{AH1}) shows that $\mathcal{L}_{X_{F}^{(0)}%
}J=\{F,J\}_{0}=-\mathcal{L}_{\Upsilon}F$ for any $F\in C^{\infty}(M)$. Using
above relations and (\ref{Hor1}), we obtain (\ref{Hor3}). Next, if $Y$ is an
Poisson vector field $Y$ of the bracket $\{,\}_{0}$, then
\[
[ Y,X_{J}^{(0)}]=X_{\mathcal{L}_{Y)}J}^{(0)}.
\]
Combining this identity for Poisson vector fields $\operatorname{hor}_{i}^{p}$
and $\operatorname{hor}_{i}^{q}$ with equalities (\ref{Hor3}), we justify
(\ref{Hor4}).
\end{proof}

Therefore, it follows from (\ref{Hor3}) and (\ref{Hor4}) that vector fields in
(\ref{Hor1}) are $\mathbb{S}^{1}$-invariant and have the momentum map $J$ as a
common first integral. The following consequence of Lemma \ref{lem2} gives us an
alternative definition of $\operatorname{hor}_{i}^{p}$ and $\operatorname{hor}%
_{i}^{q}$.

\begin{corollary}
The vector fields in (\ref{Hor1}) coincide with the $\mathbb{S}^{1}$-averages
of the coordinate vector fields associated to the slow variables,
\begin{equation}
\operatorname{hor}_{i}^{p}= \langle \frac{\partial}{\partial p^{i}} \rangle  ,\text{
\ \ }\operatorname{hor}_{i}^{q}= \langle \frac{\partial}{\partial q^{i}} \rangle   \label{HB}%
\end{equation}
for all $i=1,...,k$.
\end{corollary}

\begin{proof}
By (\ref{AH1}), the $\mathbb{S}^{1}$-action is Hamiltonian relative to the
Poisson bracket $\{,\}_{0}$ and hence for any $F\in C^{\infty}(M)$, the
$\mathbb{S}^{1}$-average of the Hamiltonian vector $X_{F}^{(0)}$ is given by
\begin{equation}
 \langle X_{F}^{(0)} \rangle  =X_{ \langle F \rangle  }^{(0)} \label{GP}%
\end{equation}
In particular, the condition $ \langle F \rangle  =0$ implies that $ \langle X_{F}^{(0)} \rangle  =0$. Then,
 it follows from (\ref{AZ}) that
\begin{equation}
\langle X_{\Theta_{i}^{p}}^{(0)} \rangle  = \langle X_{\Theta_{i}^{q}}^{(0)} \rangle  =0. \label{PG}%
\end{equation}
These relations and the $\mathbb{S}^{1}$-invariance of vector fields
(\ref{Hor1}) imply (\ref{HB}).
\end{proof}

Let us think of the domain $M\subset\mathbb{R}_{y,x}^{2r}\times\mathbb{R}%
_{p,q}^{2k}$ as the total space  of a trivial symplectic  bundle whose base
is the projection of $M$ to the ``slow" $(p,q)$-space and the fibers are given
by the intersections of $M$ with slices $\mathbb{R}_{y,x}^{2r}\times
\{(p,q)\}$. The $\mathbb{S}^{1}$-action leave invariant the fibers whose
tangent spaces form the vertical distribution $\mathbb{V}%
=\operatorname*{Span}\{\frac{\partial}{\partial y},\frac{\partial}{\partial
x}\}$. Denote by $\mathbb{H}$ the distribution on $M$ spanned by vector fields
$\operatorname{hor}_{i}^{p}$ and $\operatorname{hor}_{i}^{q}$ in (\ref{Hor1})
for $i=1,...,k$. Then, we have the $\mathbb{S}^{1}$-invariant splitting
\begin{equation}
TM=\mathbb{V}\oplus\mathbb{H}. \label{SPL}%
\end{equation}
Relations (\ref{HB}) show that the horizontal distribution $\mathbb{H}$
gives the Hannay-Berry connection in the sense of \cite{MaMoRa-90,Mon-88}. This connection is obtained by the averaging of the trivial
connection on $M$ with respect to the $\mathbb{S}^{1}$-action associated to
the periodic flow of $X_{H}^{(0)}$. The horizontal lifts of vector fields on
the base with respect to the Hannay-Berry are just given by 
(\ref{Hor1}). The vector fields tangent to the distributions $\mathbb{V}$ and
$\mathbb{H}$ are said to be vertical and horizontal, respectively. The
curvature of the Hannay-Berry connection is zero if and only if the horizontal
distribution $\mathbb{H}$ is integrable. This happens in the case when the
vector fields in (\ref{Hor1}) pairwise commute.

Now, let us consider the following $\mathbb{S}^{1}$-invariant bivector field
on $M$:
\begin{equation}
\Pi_{\Theta}:=\operatorname{hor}_{i}^{p}\wedge\operatorname{hor}_{i}^{q}
\label{HLF}%
\end{equation}
which is just the horizontal lift of the Poisson tensor on $\mathbb{R}%
_{p,q}^{2k}$ with respect to splitting (\ref{SPL}).

\begin{lemma}\label{lemm7}
The $\mathbb{S}^{1}$-average of the Poisson tensor of the slow Poisson bracket
$\{,\}_{1}$ has the representation
\begin{equation}
 \langle \frac{\partial}{\partial p}\wedge\frac{\partial}{\partial q} \rangle  =\Pi_{\Theta
}-\mathcal{L}_{ \langle V \rangle  }(\frac{\partial}{\partial y}\wedge\frac{\partial}{\partial
x}), \label{For1}%
\end{equation}
where $V$is a vector field on $M$ given by
\begin{equation}
V:=\frac{1}{2}\mathbf{i}_{\Theta}(\frac{\partial}{\partial p}\wedge
\frac{\partial}{\partial q})\equiv\frac{1}{2}\left(  \Theta_{i}^{p}%
\frac{\partial}{\partial q^{i}}-\Theta_{i}^{q}\frac{\partial}{\partial p^{i}%
}\right)  . \label{For2}%
\end{equation}
\end{lemma}

\begin{proof}
By straightforward but lengthly calculations, we verify the following
identity
\begin{align}
\frac{\partial}{\partial p}\wedge\frac{\partial}{\partial q}  &  =\Pi_{\Theta
}-\frac{1}{2}\operatorname{hor}_{i}^{p}\wedge X_{\Theta_{i}^{q}}^{(0)}%
+\frac{1}{2}\operatorname{hor}_{i}^{q}\wedge X_{\Theta_{i}^{p}}^{(0)}%
\label{For3}\\
&  -\mathcal{L}_{V}(\frac{\partial}{\partial y}\wedge\frac{\partial}{\partial
x}).\nonumber
\end{align}
The first term on the right hand side of this equality is $\mathbb{S}^{1}%
$-invariant. The $\mathbb{S}^{1}$-average of the corresponding second and
third terms is zero because  of properties (\ref{HB}) and (\ref{PG}).
Finally, taking into account that the Poisson tensor $\frac{\partial}{\partial
y}\wedge\frac{\partial}{\partial x}$ is $\mathbb{S}^{1}$-invariant and
averaging the both sides of (\ref{For3}), we get decomposition (\ref{For1}).
\end{proof}

\section{An Intrinsic Splitting of Normal Forms}

Here, we apply the (non-canonical) Lie transform method to the perturbed
Hamiltonian vector field of system (\ref{SF1}), (\ref{SF2}). Taking into
account that the normalization procedure contains a certain freedom of
formulation, we show how to fix this freedom to get an intrinsic splitting of
a first order normal form.

We say that an open domain $N$ in $M$ is admissible if its closure $\bar{N}$
is compact and invariant with respect to the $\mathbb{S}^{1}$-action. By a
near-identity transformation we mean a smooth family of mappings
$\mathcal{T}_{\varepsilon}:N\rightarrow M,$ $\varepsilon\in(-\varepsilon
_{0},\varepsilon_{0})$ such that $\mathcal{T}_{0}=\operatorname*{id}$ and
$\mathcal{T}_{\varepsilon}$ is a diffeomorphism onto its image.

\begin{theorem}\label{maintheo}
Assume that the flow of the unperturbed Hamiltonian vector field $X_{H}^{(0)}$
is periodic with frequency function $\omega$. Then, for every admissible
domain domain $N\subset M$ and small enough $\varepsilon$, there exists a near
identity transformation $\mathcal{T}_{\varepsilon}:N\rightarrow M$  which
brings the Hamiltonian vector field $X_{H}=X_{H}^{(0)}+\varepsilon X_{H}%
^{(1)}$ of slow-fast system (\ref{SF1}), (\ref{SF2}) to the following
$\mathbb{S}^{1}$-invariant normal form of first order:%
\begin{equation}
\mathcal{T}_{\varepsilon}^{\ast}X_{H}=X_{H}^{(0)}+\varepsilon
(P_{\operatorname{hor}}+P_{\operatorname{ver}})+O(\varepsilon^{2}),
\label{NFG}%
\end{equation}
where $\ $the horizontal $P_{\operatorname{hor}}$ and vertical
$P_{\operatorname{ver}}$ vector fields on $M$ are given by%
\[
P_{\operatorname{hor}}:=\mathbf{i}_{dH}\Pi_{\Theta},\text{ \ }%
P_{\operatorname{ver}}:=X_{ \langle K \rangle  }^{(0)}%
\]
and
\begin{equation}
K:=\frac{1}{2}\left(  \mathcal{S}(\frac{\partial J}{\partial p^{i}}%
)\frac{\partial H}{\partial q^{i}}-\mathcal{S}(\frac{\partial J}{\partial
q^{i}})\frac{\partial H}{\partial p^{i}}\right)  . \label{COR1}%
\end{equation}
\end{theorem}

\begin{proof}
First, let us apply a general normalization result \cite{MisVor-12} to the
perturbation vector field $X_{H}$. Let
\begin{equation}
Z=\frac{1}{\omega}\mathcal{S}(X_{H}^{(1)})+\frac{1}{\omega^{3}}\mathcal{S}%
^{2}(\mathcal{L}_{X_{H}^{(1)}}\omega)X_{H}^{(0)}+Y, \label{NAG}%
\end{equation}
where $Y$ is an arbitrary $\mathbb{S}^{1}$-invariant vector field on $M$.
Denote by
\begin{equation}
\mathcal{T}_{\varepsilon}=\operatorname{Fl}_{Z}^{t}\mid_{t=\varepsilon}
\label{FAG}%
\end{equation}
the time-$\varepsilon$ flow of the vector field $Z$. Then, for small enough
$\varepsilon$, the near-identity transformation $\mathcal{T}_{\varepsilon}$
sends $X_{H}$ to the following first order normal form \cite{MisVor-12} :
\begin{equation}
\mathcal{T}_{\varepsilon}^{\ast}X_{H}=X_{H}^{(0)}+\varepsilon\left(
 \langle X_{H}^{(1)} \rangle  +\frac{1}{\omega}\mathcal{L}_{Y}(\omega)X_{H}^{(0)}\right)
+O(\varepsilon^{2}). \label{FAG1}%
\end{equation}
Next, let us choose an $\mathbb{S}^{1}$-invariant vector field $Y$ in a such a
way that the second terms in normal forms (\ref{FAG1}) and (\ref{NFG})
coincide. It follows from representation $X_{H}^{(1)}=\mathbf{i}_{dH}%
(\frac{\partial}{\partial p}\wedge\frac{\partial}{\partial q})$ that%
\begin{equation}
 \langle X_{H}^{(1)} \rangle  =\mathbf{i}_{dH} \langle \frac{\partial}{\partial p}\wedge\frac
{\partial}{\partial q} \rangle  . \label{For4}%
\end{equation}
On the other hand, taking into account that
\[
[ \langle V \rangle  ,X_{H}^{(0)}]=[ \langle V \rangle  ,\omega\Upsilon] \rangle  =(\mathcal{L}_{ \langle V \rangle  }%
\omega)\Upsilon=\frac{1}{\omega}(\mathcal{L}_{ \langle V \rangle  }\omega)X_{H}^{(0)},
\]
by the standard properties of the Lie derivative \cite{AbMa-88}, we obtain
\begin{align*}
\mathbf{i}_{dH}\circ\mathcal{L}_{ \langle V \rangle  }(\frac{\partial}{\partial y}\wedge
\frac{\partial}{\partial x})  &  =-X_{\mathcal{L}_{ \langle V \rangle  }H}^{(0)}+[ \langle V \rangle  ,X_{H}%
^{(0)}]\\
&  =-X_{\mathcal{L}_{ \langle V \rangle  }H}^{(0)}+\frac{1}{\omega}(\mathcal{L}_{ \langle V \rangle  }%
\omega)X_{H}^{(0)}.
\end{align*}
Combining this relation with (\ref{For1}) and (\ref{For4}), we get the
following representation
\begin{equation}
 \langle X_{H}^{(1)} \rangle  =P_{\operatorname{hor}}+P_{\operatorname{ver}}-\frac{1}{\omega
}(\mathcal{L}_{ \langle V \rangle  }\omega)X_{H}^{(0)}. \label{INP1}%
\end{equation}
Using (\ref{For2}), we verify that
\begin{equation}
\mathcal{L}_{V}H=K \label{INP}%
\end{equation}
and hence $\mathcal{L}_{ \langle V \rangle  }H= \langle \mathcal{L}_{V}H \rangle  = \langle K \rangle  $. Finally, the desired
choice of $Y$ in (\ref{FAG1}) is that $Y= \langle V \rangle  $.
\end{proof}

The horizontal and vertical components of the second term in normal form
(\ref{NFG}) possess the following properties. The vertical component
$P_{\operatorname{ver}}$ is a Hamiltonian vector field relative to the slow
Poisson bracket $\{,\}_{0}$ and the function $ \langle K \rangle  $ which can be interpretated
as an $\mathbb{S}^{1}$-invariant correction of first order to the Hamiltonian
$H$ (see, Theorem 6.2). The horizontal component $P_{\operatorname{hor}}$
involves the $\mathbb{H}$ -horizontal lift $\Pi_{\Theta}$ (\ref{HLF}) of the
Poisson tensor on the $(p,q)$-space which satisfies the Jacobi identity only
in the case when the horizontal distribution $\mathbb{H}$ is integrable \cite{DaVo-08,Vor-11}.
Therefore, in general, $P_{\operatorname{hor}}$ does not inherit any natural
Hamiltonian structure from $X_{H}$ .

Moreover, from (\ref{INP1}) and (\ref{INP}), we derive the following
relationship between the $\mathbb{S}^{1}$-average $ \langle X_{H}^{(1)} \rangle  $ of the
perturbation vector field and the the second term in normal form
(\ref{NFG}).

\begin{corollary}
The averaged perturbation vector field has the representation
\begin{equation}
 \langle X_{H}^{(1)} \rangle  =P_{\operatorname{hor}}+P_{\operatorname{ver}}+gX_{H}^{(0)},
\label{Main1}%
\end{equation}
where%
\[
g=-\frac{1}{2\omega}\left(  \mathcal{S}(\frac{\partial J}{\partial p^{i}%
})\frac{\partial\omega}{\partial q^{i}}-\mathcal{S}(\frac{\partial J}{\partial
q^{i}})\frac{\partial\omega}{\partial p^{i}}\right)  .
\]

\end{corollary}

Remark, that last term on the right hand side of (\ref{Main1}) is not
Hamiltonian relative to the bracket $\{,\}_{0}$, in general.

Property (\ref{Hor3}) yields $\mathcal{L}_{\mathcal{P}_{\operatorname{hor}}%
}J=0$. Moreover, $\mathcal{L}_{\mathcal{P}_{\operatorname{ver}}}%
J=\mathcal{L}_{ \langle K \rangle  }J=-\mathcal{L}_{\Upsilon} \langle K \rangle  =0$ and $\mathcal{L}%
_{gX_{H}^{(0)}}J=g\omega\mathcal{L}_{\Upsilon}J=0$. Therefore, we arrive at
the following fact.

\begin{corollary}
The momentum map $J$ (\ref{AH2}) is a first integral of the averaged
perturbation vector field,%
\begin{equation}
\mathcal{L}_{ \langle X_{H}^{(1)} \rangle  }J=0. \label{ADI}%
\end{equation}

\end{corollary}

\begin{remark}
The momentum map $J$ is uniquely determined by condition (\ref{AV1}) up to
adding a smooth function $f=f(p,q)$. But, such a renormalization of  $J$ does
not preserve property (\ref{ADI}).
\end{remark}

\begin{remark}
Suppose that the $\mathbb{S}^{1}$-action associated to the periodic flow of
$X_{H}^{(0)}$ is free on $M$ and not necessarily trivial. Then, it follows
from (\ref{ADI}) and the periodic averaging theorem \cite{Arno-63,SanVer-07} that action $J$ (\ref{AH2}) is an adiabatic invariant of
slow-fast Hamiltonian system (\ref{SF1}), (\ref{SF2}), that is, $|
J(\operatorname{Fl}_{X_{H}}^{t}(m)-J(m)|=O(\varepsilon)$ , for $m\in N$ and
 $t\sim\frac{1}{\varepsilon}$. This is just the contents of
the classical adiabatic theorem \cite{ArKN-88,Ne-08} which is usually
formulated in the case $r=1$ and for domains of action-angle variables.
\end{remark}

The following fact can be useful in theory of semiclassical quantization of
slow-fast Hamiltonian systems \cite{Kar}.

\begin{proposition}
Under hypothesis of Theorem \ref{maintheo}, the function
\begin{equation}
F=J-\frac{\varepsilon}{\omega}\mathcal{S}(\{H,J\}_{1}) \label{INT1}%
\end{equation}
is an approximate first integral on $M$ of slow-fast Hamiltonian system
(\ref{SF1}), (\ref{SF2}) in the sense that
\begin{equation}
\mathcal{L}_{X_{H}}F=O(\varepsilon^{2}). \label{APF1}%
\end{equation}
\end{proposition}
\begin{proof}
For a function $F=F_{0}+\varepsilon F_{1}$ condition (\ref{APF1}) holds if and
only if the functions $F_{0}$ and $F_{1}$ are solutions to the following
equations
\begin{equation}
\mathcal{L}_{X_{H}^{(0)}}F_{0}=0, \label{Hom1}%
\end{equation}
\begin{equation}
\mathcal{L}_{X_{H}^{(0)}}F_{1}=-\mathcal{L}_{X_{H}^{(1)}}F_{0}. \label{Hom2}%
\end{equation}
If we put $F_{0}=J$, then (\ref{Hom1}) is satisfied because of the
$\mathbb{S}^{1}$-invariance of $H$. In terms of the infinitesimal generator
$\Upsilon$, equation (\ref{Hom2}) for $F_{1}$ is written as
\begin{equation}
\mathcal{L}_{\Upsilon}F_{1}=-\frac{1}{\omega}\mathcal{L}_{X_{H}^{(1)}}J.
\label{Hom3}%
\end{equation}
By the identity $ \langle \mathcal{L}_{X_{H}^{(1)}}J \rangle  =\mathcal{L}_{ \langle X_{H}^{(1)}%
 \rangle  }J$, the solvability condition of equation (\ref{Hom3}) just coincides with
(\ref{ADI}). Finally, equality (\ref{Pr1}) shows that a particular
solution to (\ref{Hom3}) is given by the formula $F_{1}=-\frac{1}{\omega
}\mathcal{S}(\mathcal{L}_{X_{H}^{(0)}}J_{1})=-\frac{1}{\omega}\mathcal{S}%
(\{H,J\}_{1})$.
\end{proof}

\section{An Approximate Hamiltonian Model with $\mathbb{S}^{1}$-Symmetry}

Here, under hypothesis of Theorem \ref{maintheo}, we give an alternative derivation of normal form splitting in
(\ref{NFG}) by applying a normalization procedure to the Poisson bracket (\ref{PB}) and
the Hamiltonian. In the first step, by means of a near-identity transformation
$\Phi_{\varepsilon}$, we correct original Poisson bracket (\ref{PB}) to get
an $\mathbb{S}^{1}$-invariant one. In the
second step, a canonical averaging transformation is applied to the
transformed Hamiltonian $H\circ\Phi_{\varepsilon}$.

First, we recall some facts concerning to the averaging procedure for
symplectic and Poisson structures. Consider the symplectic form associated
to Poisson bracket (\ref{PB}):
\[
\Omega=\frac{1}{\varepsilon}dp\wedge dq+dy\wedge dx.
\]
Then, by (\ref{FS1}) the $\mathbb{S}^{1}$-average of $\Omega$ is given by the
formula
\[
 \langle \Omega \rangle  =\Omega-d\Theta.
\]

\begin{lemma}
Let $N\subset M$ be an admissible domain. Then, for sufficiently small
$\varepsilon\neq0$, the $\mathbb{S}^{1}$-average $ \langle \Omega \rangle  $ is a symplectic
form on $N$. Moreover, there exists a near-identity transformation
$\Phi_{\varepsilon}:N\rightarrow M$ which is a symplectomorphism between
$\Omega$ and $ \langle \Omega \rangle  $,%
\begin{equation}
\Phi_{\varepsilon}^{\ast}\Omega= \langle \Omega \rangle   \label{Sym}%
\end{equation}

\end{lemma}

The proof of this lemma is based on the minimal coupling procedure and the
Moser homotopy method, see, for example,  \cite{DaVo-08,GLS-96,Vor-11,VorMis-12}. Here, we recall an algorithm of the
construction of $\Phi_{\varepsilon}$. Let us associate to the 1-form $\Theta$
the following $\lambda$-parameter family of 2-forms on $M$:
\begin{equation}
\delta_{\Theta}^{\lambda}:=d_{p}\Theta+d_{q}\Theta+\frac{(1-\lambda)}
{2}\{\Theta\wedge\Theta\}_{0}, \label{Curv}%
\end{equation}
where
\begin{align*}
\{\Theta\wedge\Theta\}_{0}  &  :=\{\Theta_{i}^{p}\wedge\Theta_{j}^{p}%
\}_{0}dp^{i}\wedge dp^{j}+2\{\Theta_{i}^{p}\wedge\Theta_{j}^{q}\}_{0}%
dp^{i}\wedge dq^{j}\\
&  +\{\Theta_{i}^{q}\wedge\Theta_{j}^{q}\}_{0}dq^{i}\wedge dq^{j}.
\end{align*}
Notice that the vanishing of the form $\delta_{\Theta}^{0}$ (called the
Hamiltonian 2-form of the Hannay-Berry connection) provides the integrability
of the horizontal distribution $\mathbb{H}$. Let $W_{\lambda}$ be a
time-dependent horizontal vector field on $N$ which is uniquely determined by
the equation
\begin{equation}
\mathbf{i}_{W_{\lambda}}\left(  dp\wedge dq-\varepsilon(1-\lambda
)\delta_{\Theta}^{\lambda}\right)  =-\varepsilon\Theta\label{FOM1}%
\end{equation}
Here, we use the fact: for small enough $\varepsilon$ and $\lambda\in
[0,1]$, the 2-form on the left hand side of (\ref{FOM1}) is
nondegenerate on $\bar{N}$ along the horizontal distribution $\mathbb{H}$. Then,
the symplectomorphism $\Phi_{\varepsilon}$ in (\ref{Sym}) is defined as the
time-$1$ flow of $W_{\lambda}$,%
\begin{equation}
\Phi_{\varepsilon}=\operatorname{Fl}_{W_{\lambda}}^{\lambda}\mid_{\lambda=1}.
\label{FOM2}%
\end{equation}

Denote by $\{,\}^{\operatorname*{inv}}$ the nondegenerate Poisson bracket
associated to the symplectic form  $ \langle \Omega \rangle  $ on $N$. Then,
the bracket $\{,\}^{\operatorname*{inv}}$ is $\mathbb{S}^{1}$-invariant and has the
decomposition
\begin{equation}
\{F,G\}^{\operatorname*{inv}}=\{F,G\}_{0}+\varepsilon\Pi_{\Theta
}(dF,dG)+O(\varepsilon^{2}), \label{FOM3}%
\end{equation}
where the bivector field $\Pi_{\Theta}$ is given by (\ref{HLF}).
\begin{theorem}\label{theo6.2}
Under the hypothesis of Theorem \ref{maintheo}, for small enough $\varepsilon$, there exists a near-identity transformation
$\mathcal{\tilde{T}}_{\varepsilon}:N\rightarrow M$ with the following properties:
\begin{itemize}
\item[(a)] $\mathcal{\tilde{T}}_{\varepsilon}$ is a Poisson isomorphism between
Poisson brackets $\{,\}^{\operatorname*{inv}}$ and $\{,\}$;

\item[(b)] the transformed Hamiltonian is of the form%
\[
H\circ\mathcal{\tilde{T}}_{\varepsilon}=H+\varepsilon \langle K \rangle  +O(\varepsilon^{2}),
\]
\end{itemize}
where $K$ is just given by (\ref{COR1}).
\end{theorem}
\begin{proof}
Applying transformation (\ref{FOM2}) to the original Hamiltonian system
$(M,\{,\},H)$, we get the following one%
\begin{equation}
(N,\{,\}^{\operatorname*{inv}},H\circ\Phi_{\varepsilon}=H+\varepsilon
H_{1}+O(\varepsilon^{2})),\label{PM}%
\end{equation}
where the correction term $H_{1}$ is not necessarily $\mathbb{S}^{1}%
$-invariant. To put the Hamiltonian in (\ref{PM}) to an $\mathbb{S}^{1}%
$-invariant normal form of first order, we apply a canonical transformation
defined as the time-$\varepsilon$ flow of the Hamiltonian vector
field $\tilde{X}_{G}$ relative to the Poisson bracket
$\{,\}^{\operatorname*{inv}}$ and a function $G$ which satisfies the
homological equation
\begin{equation}
\{H,G\}^{\operatorname*{inv}}=H_{1}- \langle H_{1} \rangle  \text{.}\label{HOE}%
\end{equation}
In terms of the integrating operator, a particular solution to this equation
is represented as $G=\mathcal{S}(H_{1})$. Expanding equation (\ref{FOM1}) and
the transformation $\Phi_{\varepsilon}$ at $\varepsilon=0$,
we show that the $\mathbb{S}^{1}$-averages of \ $H_{1}$ and $K$ (\ref{COR1})
coincide, $ \langle H_{1} \rangle  = \langle K \rangle  $. Finally, we conclude that the desired normalization
transformation is defined as the composition
\begin{equation}
\mathcal{\tilde{T}}_{\varepsilon}=\Phi_{\varepsilon}\circ\operatorname{Fl}%
_{\tilde{X}_{\mathcal{S}(H_{1})}}^{\varepsilon}.\label{NTN}%
\end{equation}
\end{proof}

Consider the model $\mathbb{S}^{1}$-invariant Hamiltonian system
\begin{equation}
(N,\{,\}^{\operatorname*{inv}},H+\varepsilon \langle K \rangle  ) \label{HM}%
\end{equation}
Property (\ref{Hor4}) imply that the infinitesimal generator $\Upsilon$ is
Hamiltonian relative to the Poisson bracket $\{,\}^{\operatorname*{inv}}$
and $J$, $\Upsilon=\tilde{X}_{J}.$
Therefore, the $\mathbb{S}^{1}$-action is Hamiltonian
on $(N,\{,\}^{\operatorname*{inv}})$ with momentum map $J$. It follows that the
truncated Hamiltonian $H+\varepsilon \langle K \rangle  $ and $J$ Poisson commute. For small
$\varepsilon$, these functions are independent if $H$ satisfies condition
(\ref{Cr}).

\begin{corollary}
Normalization transformation (\ref{NTN}) carries the original slow-fast
Hamiltonian system $(M,\{,\},H)$ into a system which is $\varepsilon^{2}%
$-close to the Hamiltonian model with $\mathbb{S}^{1}$-symmetry (\ref{HM}).
\end{corollary}

Moreover, an easy verification, by using (\ref{FOM3}), shows that the first
order term in the Taylor expansion at $\varepsilon=0$ of the Hamiltonian
vector field of system (\ref{HM}) coincides with the normal form of first
order in (\ref{NFG}).

\begin{remark}
In fact, by the standard Deprit normal form argument \cite{Cush-84,Dep-69} and by the
fact that homological equation of the type (\ref{HOE}) is solvable, one can
extend $\mathcal{\tilde{T}}_{\varepsilon}$ to a normalization transformation
of arbitrary order $n\geq2$  in $\varepsilon$. This means that one can
correct formula (\ref{INT1}) to get an approximate first integral of system
(\ref{SF1}), (\ref{SF2}) which satisfies condition  (\ref{APF1})
$\operatorname*{mod}O(\varepsilon^{n})$.
\end{remark}
\begin{remark}
Theorem \ref{maintheo} and Theorem \ref{theo6.2} carry a general character and can be directly
generalized to a class of slow-fast Hamiltonian systems on a phase space
$M=M_{0}\times M_{1}$ which is the product of an exact (fast) symplectic
manifold $M_{0}$ and an arbitrary (slow) symplectic manifold $M_{1}$ (see, also \cite{Vor-11,VorMis-12}).
\end{remark}

\section{The Quadratic Case}

To illustrate our general results, we consider the particular case when $r=1$
and the Hamiltonian $H$ is a quadratic function in the fast variables
$\mathbf{z}=(y,x)\in\mathbb{R}^{2}$. Let us associated to every
matrix-valued function $\mathbf{A}\in\operatorname*{sl}(2,\mathbb{R})\otimes
C^{\infty}(\mathbb{R}_{p,q}^{2k})$ the following function
\[
Q_{\mathbf{A}}=-\frac{1}{2}\mathbf{JAz}\cdot\mathbf{z},
\]
where $\mathbf{J}=\left(
\begin{array}
[c]{cc}%
0 & -1\\
1 & 0
\end{array}
\right)  $. Then, the Hamiltonian vector field relative to the bracket
$\{,\}_{0}$ and $Q_{\mathbf{A}}$ is given by $X_{Q_{\mathbf{A}}}%
^{(0)}=\mathbf{Az}\cdot\frac{\partial}{\partial\mathbf{z}}$. Consider
slow-fast Hamiltonian system (\ref{SF1}), (\ref{SF2}) with Hamiltonian of the
form
\[
H=h+\omega Q_{\mathbf{A}}%
\]
for some smooth functions $h=h(p,q),$ $\omega=\omega(p,q) >  0$ and
$\mathbf{A}=\mathbf{A}(p,q)\in\operatorname*{sl}(2;\mathbb{R})$. We assume
that $\det\mathbf{A}=1$ on a certain open domain in $\mathbb{R}_{p,q}^{2k}$.
Then, the flow of $X_{Q_{\mathbf{A}}}^{(0)}$ is periodic with frequency
function $\omega$ and the associated $\mathbb{S}^{1}$-action is given by the
linear $2\pi$-periodic flow $\operatorname{Fl}_{\Upsilon}^{t}=\cos
t\mathbf{I+}\sin t\mathbf{A}$. The corresponding momentum map (\ref{AH2}) is
$J=Q_{\mathbf{A}}$. In this case, operators in (\ref{AV1}) and (\ref{AV2})
possess the following properties.
\begin{lemma}
For arbitrary $\mathbf{B},\mathbf{C}\in\operatorname*{sl}(2,\mathbb{R})\otimes
C^{\infty}(\mathbb{R}_{p,q}^{2k})$, the following identities hold%
\begin{equation}
\langle Q_{\mathbf{B}} \rangle  =\frac{1}{2}Q_{\mathbf{B}-\mathbf{ABA}}, \label{Q1}%
\end{equation}%
\begin{equation}
\mathcal{S}(Q_{\mathbf{B}})=\frac{1}{4}Q_{[\mathbf{A},\mathbf{B}]}, \label{Q2}%
\end{equation}%
\begin{align}
&   \langle Q_{\mathbf{B}}Q_{\mathbf{C}} \rangle  =\frac{1}{4}Q_{\mathbf{B}-\mathbf{ABA}%
}Q_{\mathbf{C}-\mathbf{ACA}}\label{Q3}\\
&  +\frac{1}{8}Q_{\mathbf{B}+\mathbf{ABA}}Q_{\mathbf{C}+\mathbf{ACA}}+\frac
{1}{8}Q_{[\mathbf{B},\mathbf{A}]}Q_{[\mathbf{C},\mathbf{A}]}.\nonumber
\end{align}

\end{lemma}

Using the identities $\mathbf{A}^{-1}=-\mathbf{A}$ and $\mathbf{JA=-A}%
^{T}\mathbf{J}$, one can verify identities (\ref{Q1})-(\ref{Q3}) by a direct
computation. As a consequence of (\ref{Q2}), we get that the components of
1-form $\Theta$ in (\ref{Hor2}) are given by the formulas%

\[
\Theta_{i}^{p}=\frac{1}{2}Q_{\mathbf{A}\frac{\partial\mathbf{A}}{\partial
p^{i}}},\text{ \ }\Theta_{i}^{q}=\frac{1}{2}Q_{\mathbf{A}\frac{\partial
\mathbf{A}}{\partial q^{i}}}.
\]
In this case, it is easy to see that $\delta_{\Theta}^{0}=0$ and hence the
curvature of the Hannay-Berry connection is zero. Combining above relations
with (\ref{Q1}), we show that the $\mathbb{S}^{1}$-invariant function $ \langle K \rangle  $ in
(\ref{COR1}) is represented as follows
\[
 \langle K \rangle  =\frac{\omega}{4}\left(  Q_{\mathbf{A}\frac{\partial\mathbf{A}}{\partial
p^{i}}}Q_{\frac{\partial\mathbf{A}}{\partial q^{i}}}-Q_{\mathbf{A}%
\frac{\partial\mathbf{A}}{\partial q^{i}}}Q_{\frac{\partial\mathbf{A}%
}{\partial p^{i}}}\right).
\]
Finally, an easy computation by using (\ref{Q2}) and (\ref{Q3}) shows that the
approximate first integral $F$ (\ref{INT1}) is written in the form
\[
F=Q_{\mathbf{A}}-\frac{\varepsilon}{4\omega}\left(  Q_{[\mathbf{A}%
,\mathbf{B}]}+Q_{\mathbf{A}}Q_{[\mathbf{A},\mathbf{C}]}\right)  ,
\]
where%
\[
\mathbf{B}:=\{h,\mathbf{A}\}_{1}, \ \ \mathbf{C:}=\{\omega
,\mathbf{A}\}_{1}%
\]
and the Poisson bracket between $h$ and $\mathbf{A}$ is defined entry by entry.

\textbf{ACKNOWLEDGMENTS.} We thank M.V.Karasev and J. A. Vallejo for helpful discussions.

\end{document}